\documentclass[conference, a4paper]{IEEEtran}
\IEEEoverridecommandlockouts

\usepackage[T1]{fontenc}
\usepackage{multirow,array}
\usepackage{amsmath}
\usepackage{amsthm} 
\usepackage{amssymb}
\usepackage{subcaption}
\usepackage{algpseudocode,algorithm,algorithmicx}
\usepackage{flushend}

\usepackage[a4paper]{geometry}
\newcommand{\ie}{{\em i.e.,}\ }
\newcommand{\eg}{{\em e.g.,}\ }

\newcommand{\cP}	{{\mathcal P}}

\usepackage{graphicx}
\usepackage[colorinlistoftodos]{todonotes}

\newcommand\blfootnote[1]{%
  \begingroup
  \renewcommand\thefootnote{}\footnote{#1}%
  \addtocounter{footnote}{-1}%
  \endgroup
}

\newtheorem{theorem}{Theorem}
\newtheorem{definition}{Definition}
\newtheorem{property}{Property}

\theoremstyle{remark}

\title{A Paradigm For Collaborative Pervasive Fog Computing Ecosystems at the Network Edge}

\author{Abderrahmen Mtibaa\\
 Computer Science Department\\
University Of Missouri--St. Louis\\
amtibaa@umsl.edu}

\begin{document}
\maketitle

\begin{abstract}
\blfootnote{\color{red}This paper has been accepted by the IEEE International Conference on Communication, Networks and Satellite (COMNETSAT) 2023. The definite version of this work
will be published by IEEE as part of the COMNETSAT conference proceedings.}While the success of edge and fog computing increased with the proliferation of the Internet of Things (IoT) solutions, such novel computing paradigm, that moves compute resources closer to the source of data and services, must address many challenges such as reducing communication overhead to/from datacenters, the latency to compute and receive results, as well as energy consumption  at the mobile and IoT devices. 
fog-to-fog (f2f) cooperation has recently been proposed to increase the computation capacity at the network edge through cooperation across multiple stakeholders. 
In this paper we adopt an analytical approach to studying f2f cooperation paradigm. We highlight the benefits of using such new paradigm in comparison with traditional three-tier fog computing paradigms. We use a Continuous Time Markov Chain (CTMC) model for the $N$ f2f cooperating nodes and cast cooperation as an optimization problem, which we solve using the proposed model. 

\noindent\bf Fog computing; cooperation; multi-tenant; fog-to-fog
\end{abstract}

\section{Introduction}
\label{sec.intro}


The proliferation of hardware and software technology advancements has pushed services and computations towards the network edge in order to reduce energy consumption, delay, and core network overhead~\cite{mtibaa13cloudcom, cloudlet09}. The fog computing paradigm brings together storage, communication, and computation resources closer to users' end-devices. Therefore, fog nodes are deployed at the edge of the network, offering low latency access to users. While most deployed systems adopt a three-tier architecture, consisting of always probe an edge node before sending any task to the distant cloud, recently researchers have proposed cooperative fog layers that allow fog-to-fog cooperation and reduce the probability to probe the distant cloud. Fog nodes are expected to be deployed both hierarchically, and horizontally. Nodes in the same layer can cooperate one with each other to reduce the communication latency associated to reach higher levels.

\noindent{\bf Why To Cooperate?}
Resources at the edge will increase in the upcoming years while demand may also increase at a higher rate which makes resource provisioning a very challenging task. Cooperation and sharing resources offer solutions to such issue. Two fog nodes cooperate by sharing their resources to satisfy their clients in case of local and transient overload. Such cooperation can be incentivized via a credit-based~\cite{Bera1812:CICO} or a tit-for-tat mechanism which allows sharing resources with the expectation to use remote resources back within a ``short term''. However, leveraging such sparsely distributed muli-tenant, multi-stakeholder computing resources across the edge network is a very challenging task.

\begin{figure}[!h]
\centering
\includegraphics[width=0.85\linewidth]{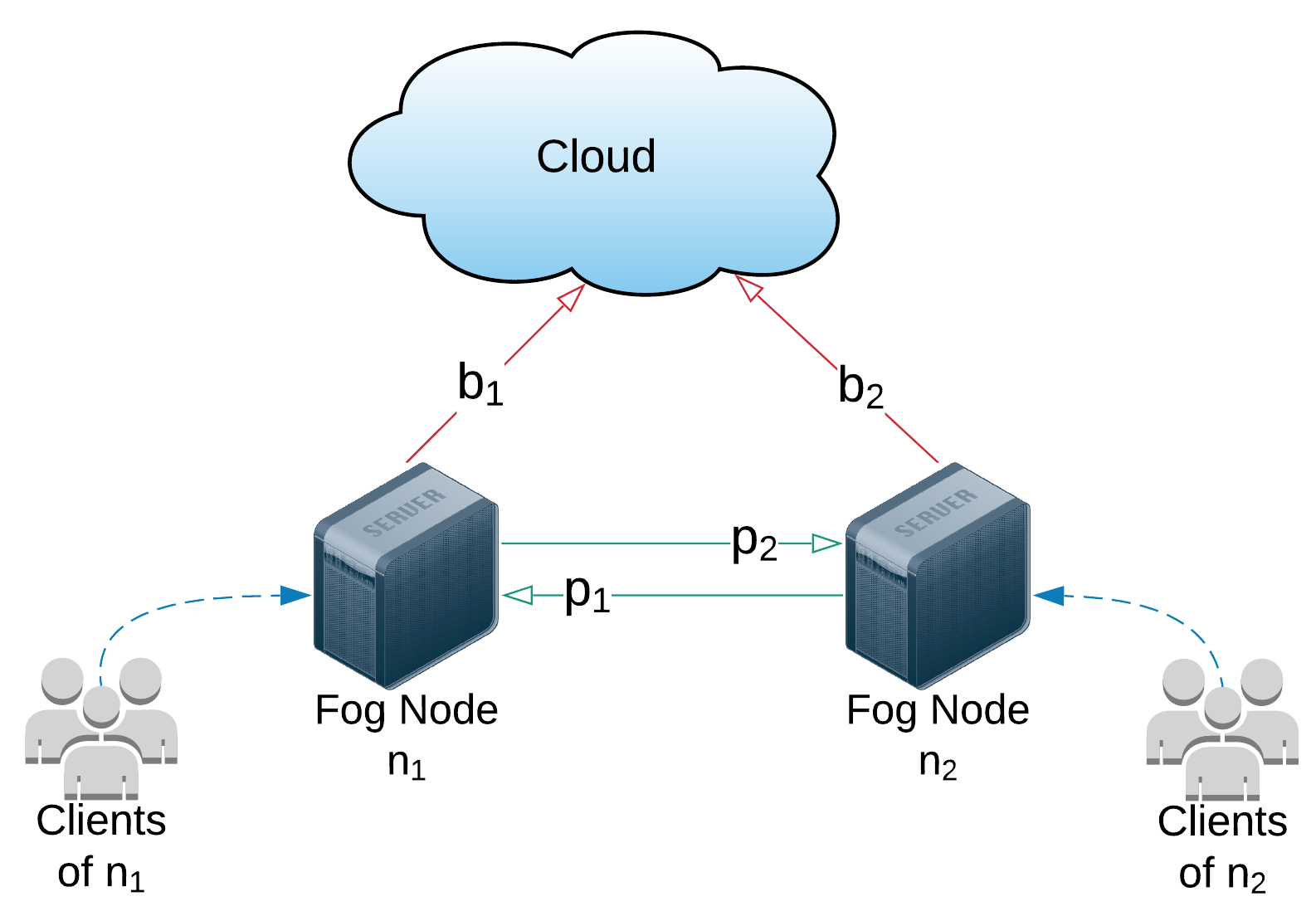}
\caption{Cooperation among two fog nodes receiving tasks from their corresponding clients; fog node $n_1$ accepts tasks from neighboring nodes with a cooperation probability $p_1$ and sends tasks to the cloud with a blocking probability $b_1$.}
\label{fig:scenario}
\vspace{-0.4cm}
\end{figure}

In this paper, we consider the following scenario of $N$ fog nodes belonging to different service providers interconnected in a mesh setting. Each fog node is programmed to compute a task if its resources are idle, otherwise offload the task to its cloud resource. This paper proposes provisioning resources from neighboring fog nodes and enabling a fog-to-fog (f2f) cooperative scheme that aims at reducing the task execution delay and the overhead at the core of the network. Fig.~\ref{fig:scenario} illustrates a cooperation example between two fog nodes $n_1$ and $n_2$. These fog nodes are directly connected to their corresponding clients which send tasks for execution. 

We define the {\em cooperation probability} $p_1$ as the probability that node $n_1$ will accept remote tasks from neighboring nodes (\eg tasks received by $n_2$ from its corresponding clients) to run. We assume that fog nodes receiving tasks from their corresponding clients will always probes other neighboring fog nodes if they cannot execute tasks locally before forwarding such tasks to the cloud. We define the blocking probability, $b_1$ at node $n_1$, as the ratio of tasks unable to be executed at the fog layer (local fog nodes and neighboring nodes) and must be forwarded to the distant cloud. The goal is to reduce such blocking probability such as we reduce the task execution delay, the overhead at the core of the network, and energy consumption of the IoT clients~\cite{cloudlet09, MAUI}. The above scenario can be expanded to $N$ fog nodes, where fog node $n_j$ cooperates with probability $p_j$ and probe another node uniformly at random.

\noindent\textbf{When To Cooperate?} 
We propose an f2f cooperation scheme under fair load distribution constraint across different fog nodes/providers. We define an optimization problem that ensures a positive gain of cooperation at all fog nodes while maintaining remote task acceptance fairness as follows: 
\begin{eqnarray}
\textbf{Problem $\cP$:} & \nonumber  \\
\underset{\mathbf{p}} {\text{Minimize:}}
 &  B(\mathbf{p}) = \left[ b_1(\mathbf{p}),\ldots,b_N(\mathbf{p}) \right ]
 \label{prolem:min}
 \\
\text{Subject to:} \nonumber \\
&  b_i(\mathbf{p}) \le b^0_i \label{problem:convenience}\\
&   R_{i}^{in} = R_{i}^{out} \label{problem:fair}\\
& 0 \le p_i \le 1 \forall i = 1,..,N \label{problem:feasable}
\end{eqnarray}
where $\mathbf{p}=(p_1,\ldots,p_N)$ is a vector of cooperation probabilities ( $p_i$ being the probability that node $i$ cooperates), $b_i(\mathbf{p})$ is the blocking probability at node $i$ when nodes cooperate with the probabilities in $\mathbf{p}$, $b^0_i$ is the blocking probability at node $i$ when nodes do not cooperate, $R_i^{in}$ and $R_i^{out}$ are the the average number of accepted tasks to run at node $i$ and the number of tasks sent for remote execution at another fog node $j\neq i$ respectively. 

Our main contributions include: 
\begin{enumerate}
    \item Making the case for compute cooperation across different fog nodes/providers. We propose optimization problem that ensures a positive gain of cooperation at all fog nodes. We highlight the potential gains of such cooperative edge computing ecosystem.
    \item We propose a simple yet general Continuous Time Markov Chain to model $N$ fog-to-fog (f2f) competitive fog providers.
    \item We solve Problem $\cP$ using the proposed model, in a closed form for $N=2$ and numerically for $N>2$
\end{enumerate}

The remainder of the paper is divided as follows. We present the related work in Section~\ref{sec.rw}. Section~\ref{sec.markov} discussed our f2f cooperation models for $N=2$ and $N>2$ nodes and provides numerical analysis of these proposed models. 
We conclude and discuss limitations and future work in Section~\ref{sec.conc}.

\section{Related Work}
\label{sec.rw}

In the literature, fog-to-fog (f2f) cooperation is recently getting attention of multiple stakeholders to improve cost-performance trade-off. However, cooperation among fog nodes were largely limited to cooperation within the same provider/stakeholder~\cite{Xiao2017Infocom, Masri2017MinimizingDI, Kapsalis2017ACF}.
Collaborative offloading schemes of unprocessed workload are proposed to reduce end-to-end latency~\cite{Masri2017MinimizingDI} and the Quality of Experience (QoE) of edge computing registered users~\cite{Xiao2017Infocom}. Within the context of the same providers, some have proposed schemes requiring periodic exchange of control messages to a central node, or a broadcast to all nodes to make efficient compute and scheduling decisions~\cite{Kapsalis2017ACF}. 

While most of researchers have investigated cooperation from within the same provider or stakeholder, few have been interested in cooperation across a federation of edge networks~\cite{Kapsalis2017ACF, Cosimo2018,  Bera1812:CICO}. In \cite{Kapsalis2017ACF}, authors propose a scheme that characterizes tasks according to their computational nature and subsequently allocates them to the appropriate hosts in the federation via a brokering Publish/Subscribe asynchronous communication system. Another study introduced Fog Infrastructure Providers (FIPs) to mutually sharing workloads and resources. Authors show that the coalition among cooperative FIPs improve their net profit. Incentive mechanisms for collaboration across providers has been proposed and discussed in~\cite{Bera1812:CICO}. However, most of these proposed frameworks require exchange of rates across fog nodes to efficiently allocate sub-intervals of task executions. Differently from other  works, we propose a novel model that does not rely on load prediction and efficiently tune the cooperation among the nodes. 

The work in \cite{OurTCC} reports a study on the cooperation among different fog nodes with the purpose of load balancing. Cooperation probabilities are there used to obtain fairness among nodes. The model adopted in the work is valid for $N \rightarrow \infty$ and the purpose of cooperation is different from ours. Authors empirically found the optimal cooperating probabilities, but without framing the problem into an optimization framework. 

Finally, researchers have studied cooperation among cloud providers~\cite{TruongHuu2014ANM}, however proposed solutions often require centralized cloud controller and cannot operate in a distributed and decentralized fashion.

\section{Fog-to-Fog Cooperation}
\label{sec.markov}


\subsection{Motivation}
Cooperation across different providers may be counter-intuitive as providers do not want to help competitors achieve good performances. We will show that this selfish strategy results in a lose-lose case. By setting their cooperation probability to minimum, nodes will fast detect unfair load sharing and cancel cooperation. However, if nodes cooperate fairly following the model $\cP$, their blocking probability will decrease, which results in an increase of the QoE of the corresponding clients.

We define the overloaded state in its simplest form, namely a node $i$ is {\em overloaded} at time $t$ when $i$ is busy executing a task at time $t$. A node is {\em idle} if it is not overloaded. But other definitions are possible as well and may apply follow similar models  as the one proposed in this paper. For example, a node implements a queue of tasks waiting for being served, and the overloaded state may correspond to a queue length higher than a given threshold.



\subsection{Cooperative Model for N = 2 Nodes}
We start by considering a model of two cooperating nodes. We assume that the communication delay among fog nodes is negligible compared to fog-to-cloud (f2c) communication delay~\cite{cloudlet09}, tasks arrival is a Poisson process with mean $\lambda$ and execution time is exponentially distributed with unitary mean. 

f2f cooperation is regulated via a {\it cooperation} probability: when $n_1$ is overloaded and receives a new task, it asks $n_2$ to use its resources for remote task execution. If $n_2$ is idle it accepts to share its resources with $n_1$ with probability $p_2$. Such cooperation reduces the blocking probability at the fog nodes and thus sending tasks to the distant cloud.

The Continuous Time Markov Chain (CTMC) shown in Fig.~\ref{fig:MC} describes the dynamic of the two cooperating nodes. The state of the chain is a pair $(s_1,s_2) \in {0,1} \times {0,1}$, where $s_i$ represents the state of node $i$, \ie $s_i=0$ (or $s_i=1$) denotes that the resource of node $n_i$ is idle (or overloaded). Node $n_i$ changes its state from $s_i=0$ to $s_i=1$ if it receives one task from its own clients (tasks are received with rate $\lambda_i$), or from the neighboring fog node $n_j$ which has received a task from its clients while overloaded, \ie $s_j=1$ and $n_i$ accepts to cooperate (this occurs with rate $p_i\lambda_j)$. The dead rates of the chain are $\mu=1$. 
The blocking probability can be expressed as a function of the steady state probabilities of the chain as follows:
$$
b_1 = \pi_{11}+\pi_{10}(1-p_2) \:\:\:\:\: b_2 = \pi_{11}+\pi_{01}(1-p_1), 
$$
where the first term represents the probability that when a task arrives, the two nodes are overloaded, while the second term represents the probability that the node receiving the task is overloaded while the neighboring node is idle but does not accept to cooperate. 

\begin{figure}[tb]
\centering
\includegraphics[width=0.55\linewidth]{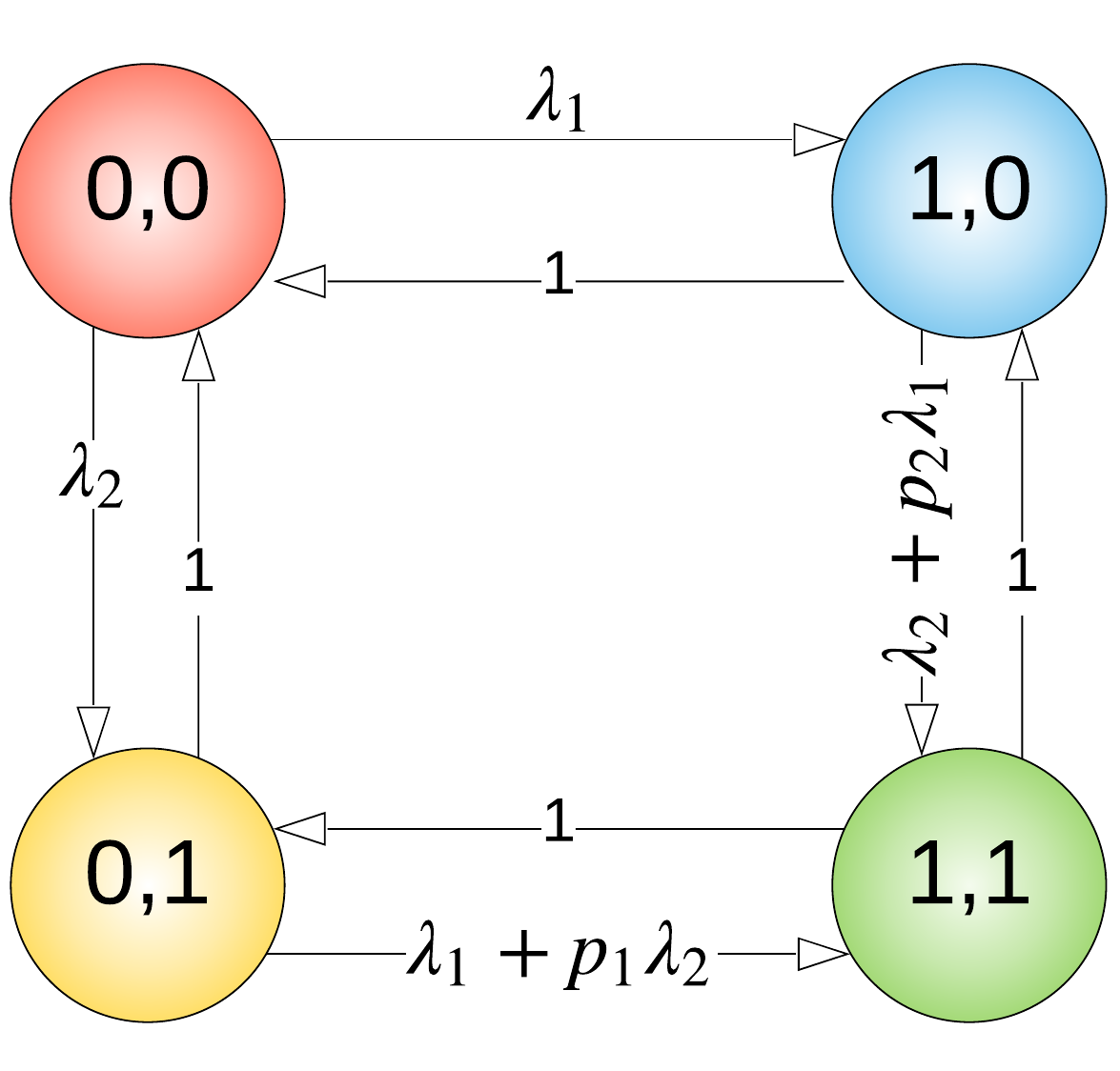}
\caption{The state transition diagram of the CTMC used to model a two node cooperating network.}
\label{fig:MC} 
\vspace{-0.4cm}
\end{figure}

The steady state probabilities can be obtained solving the CTMC. In particular, as the infinitesimal generation matrix of the CTMC is given by:

\scriptsize
$$
\mathbf{Q} = 
 \begin{bmatrix}
  -\lambda_1-\lambda_2 & \lambda_1 & \lambda_2 & 0 \\
  1 & -1-p_2\lambda_1 - \lambda_2 & 0 & \lambda_2 + p_2\lambda_1  \\
  1 & 0 & -1-p_1\lambda_2  - \lambda_1  & \lambda_1 + p_1 \lambda_2  \\
  0 & 1 & 1 & -2 
 \end{bmatrix}
$$
\normalsize 

The steady state probabilities $\pi$ are the solution of a standard system of linear equations $\mathbf{A}\pi=\mathbf{b}$, where $\mathbf{A}=\mathbf{Q}^T$, but the last row equal to all ones and $\mathbf{b}^T=[0,0,0,1]$, where the unitary entry takes the probability normalization condition into account. The Cramer's rule applied to the system allows to write:
\begin{equation}
\label{equa:sol_pi}
\pi_i = \frac{det(A_i)}{det(A)}    
\end{equation}
where $i$ is expressed in binary digits and column/row indexes start from zero. After some algebraic manipulation, the blocking probability can be conveniently expressed as a function of $p_1,p_2$:
\begin{equation}
b_1(p_1,p_2)= \frac{\kappa_1+\alpha_1p_1+\beta_1 p_2-\gamma_1 p_1p_2}{\kappa +\alpha p_1+\beta p_2+\gamma p_1p_2},
\label{formula:b1}
\end{equation}

where the coefficients are defined as:
\begin{equation}
 \left.\begin{aligned}
        \kappa &= 2+ 3\lambda_1+3\lambda_2 +4\lambda_1\lambda_2+ \lambda^2_1+\lambda^2_2+\lambda^2_1\lambda_2+\lambda^2_2\lambda_1 \\
\alpha & =\lambda_2+\lambda_1\lambda_2+2\lambda^2_2+\lambda^3_2+\lambda^2_2\lambda_1 \\
\beta & =\lambda_1+\lambda_2\lambda_1+2\lambda^2_1+\lambda^3_1+\lambda^2_1\lambda_2 \\
\gamma & 
=\lambda^2_1\lambda_2+\lambda^2_2\lambda_1\\
\kappa_1 & =2\lambda_1+3\lambda_1\lambda_2+\lambda^2_1+\lambda^2_1\lambda_2+\lambda^2_2\lambda_1 \\        \alpha_1 & =2\lambda^2_2+\lambda_1\lambda_2+\lambda_1\lambda^3_2+\lambda^3_2 \\
\beta_1 & =\lambda_2+\lambda^2_1 \lambda_2+\lambda^3_1-2\lambda_1-\lambda_1\lambda_2\\
 \gamma_1 & = \lambda_1\lambda_2+\lambda^2_2 - \lambda^2_1\lambda_2-\lambda^2_2\lambda_1        \end{aligned}
 \right.
\end{equation}
Note that all the coefficients are always positive.  
Due to the symmetry of the model, we can express $b_2$ as follows: 
$$
b_2(p_1,p_2)=b_1(p_2,p_1)
$$
where the coefficients in $b_1$ have now $\lambda_1$ and $\lambda_2$ swapped.


\subsection{Closed form solution of $\cP$ for $N=2$ Nodes}
To better illustrate the properties of the cooperation problem, we present two basic definitions concerning the minimization of $N$
multivariate functions $B(\mathbf{p})=[b_1(\mathbf{p}),..,b_N(\mathbf{p})]$ with decision variables $\mathbf{p}$, whose solution belongs to the so-called efficient set.
\begin{definition}[Pareto improvement]
A vector $\delta \mathbf{p}=(\delta p_1, \ldots \delta p_N)$ is a Pareto over $\mathbf{p}$ if
$$
b_i(\mathbf{p}+ \delta \mathbf{p}) \le b_i(\mathbf{p}) \:\:\:\: \forall i \in 1,\ldots,N
$$
$$
\exists j \in 1,\ldots N \:\:\:\: b_j(\mathbf{p}+ \delta \mathbf{p})<b_j(\mathbf{p})
$$
\end{definition}

\begin{definition} [Efficiency]
A cooperation vector $\mathbf{p}$ is efficient if it does not exit any Pareto improvement for that vector. $\partial P$ denotes the set of the efficient cooperation vectors. 
\end{definition}
Let now focus on the case $N=2$ and use solution of the MC of Fig.~\ref{fig:MC} for the closed form expression of $b_1,b_2$. We first consider a relaxed problem without the constrains (\ref{problem:convenience}),(\ref{problem:fair}). Constrains will be added later. We have the following property.


\begin{property}
\label{property1}
The efficient solution set of the function $B=[b_1,b_2]$, where $b_i$ are the solution of the $MC$ model of Fig.~\ref{fig:MC} and $p_i \in [0,1]$, is:
$$
\partial P = \{(1,*)\}\cup \{(*,1)\}
$$
\end{property} 

\begin{proof}
We will show that a Pareto improvement $ (\delta p_1, \delta p_2)$ exits for any cooperation probability pair, $(p_1,p_2)$, and it is such that $\delta p_1,\delta p_2>0$, \ie both cooperation probabilities must increase. 
Hence, a given cooperation pair can be Pareto efficient if and only if at least one component cannot be increased. 
As $b_i$ is a multivariate function, the variation $\delta p_1,\delta p_2$ that lets the function decreases is such that $\nabla b_i \cdot (\delta p_1,\delta p_2) < 0$. The blocking probability $b_1$ can be rewritten as:
\begin{eqnarray*}
b_1(p_1,p_2)&=&\frac{(\kappa_1+\beta_1p_2)+(\alpha_1+\gamma_1p_2)p_1}{(\kappa+\beta p_2)+(\alpha+\gamma p_2)p_1}
\\
&=&\frac{(\kappa_1+\alpha_1p_1)+(\beta_1+\gamma_1p_1)p_2}{(\kappa+\alpha p_1)+(\beta+\gamma p_1)p_2}
\end{eqnarray*}
Applying the derivative rule, we get: 
\begin{equation*}
\label{DERIVATA}
\frac{d}{d x} \Big ( \frac{A+Bx}{C+Dx}\Big )= \frac{BC-AD}{(C+Dx)^2}
\end{equation*}
We have:
\scriptsize
\begin{eqnarray*}
\frac{\partial b_1}{\partial p_1} &=& \frac{(\alpha_1+\gamma_1p_2)(\kappa+\beta p_2)-(\alpha+\gamma p_2)(\kappa_1+\beta_1p_2)}{(\kappa +\alpha p_1+\beta p_2+\gamma p_1p_2)^2} >0
\end{eqnarray*}
\normalsize
In fact the denominator is positive and after some manipulation the numerator can be written as:
$       \lambda_2 (2 \lambda_2 + ...) 
  (2  +... - p_2 - \lambda_2 p_2)>0
$ which is also positive. Similarly
\scriptsize
$$
\frac{\partial b_1}{\partial p_2} = \frac{(\beta_1+\gamma_1p_1)(\kappa+\alpha p_1)-(\beta+\gamma p_1)(\kappa_1+\alpha_1p_1)}{(\kappa +\alpha p_1+\beta p_2+\gamma p_1p_2)^2} < 0
$$
\normalsize

The numerator can be rewritten as:
$$
- (2 \lambda_1 + \lambda_1^2 + ..)(2 + 3 \lambda_1 +..)<0
$$
Thus, by swapping the indexes we also get:
$$
\frac{\partial b_2}{\partial p_2}>0, \frac{\partial b_2}{\partial p_1}<0
$$

Now, the two gradients $
\nabla b_1 = (\frac{\partial b_1 }{\partial p_1},\frac{\partial b_1 }{\partial p_2}),
\nabla b_2 = (\frac{\partial b_2 }{\partial p_1},\frac{\partial b_2 }{\partial p_2})
$
are not orthogonal and their angle is greater than $\pi$ since: 
$$
\nabla b_1 \cdot \nabla b_2 = \frac{\partial b_1 }{\partial p_1}  \frac{\partial b_2 }{\partial p_1} + \frac{\partial b_1 }{\partial p_2} \frac{\partial b_2 }{\partial p_2}<0
$$
Hence, a vector $\mathbf{\delta p}=(\delta p_1, \delta p_2)$, such that $\nabla b_1 \cdot \mathbf{\delta p} <0, \nabla b_2 \cdot \mathbf{\delta p} < 0$, i.e., that reduces both $b_i$ exits. And, as the angle between the gradients is higher than $ \pi$, $\delta p_1,\delta p_2 > 0$, see Fig.~\ref{fig:IMPROVE}.
\end{proof}

\begin{figure}[tbh] 
\centering
  \includegraphics[width=0.79\linewidth]{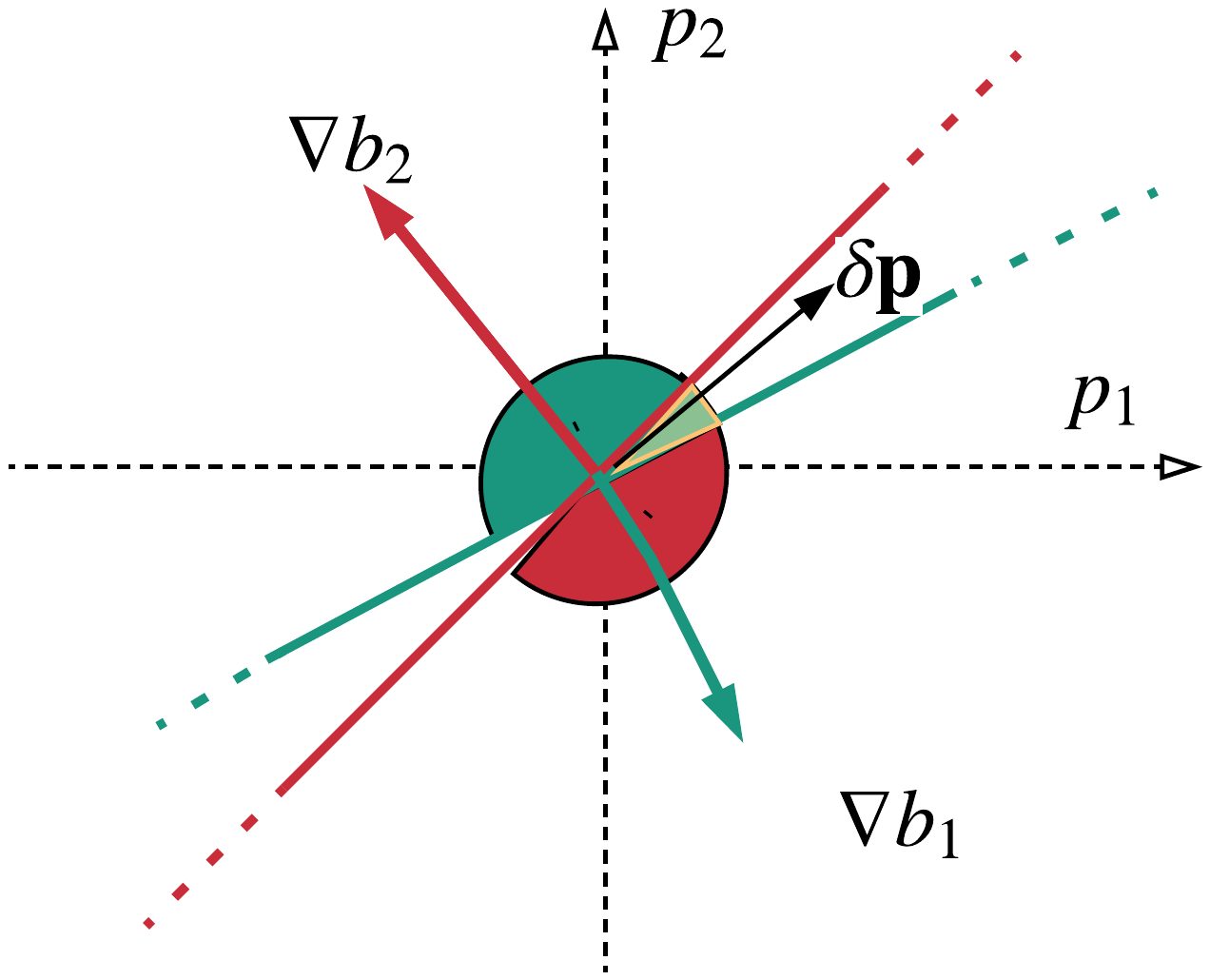}
   \caption{A sketch of a Pareto improvement; Segments are contour lines of $b_i$ and semicircles represent regions where the value of blocking probability decreases. Vectors in the intersection of the two semicircles, as $\delta \mathbf{p}$, is a Pareto improvement.}
    \label{fig:IMPROVE}  \vspace{-0.5cm}
\end{figure}



Let now find the expression of the constrain given in (\ref{problem:fair}) of $\cP$. 
Node $n_1$ executes a task from $n_2$ when: (i) an arriving task at $n_2$ finds $n_2$ congested and $n_1$ idle (this event occurs with probability $\pi_{01}$) and (ii) $n_1$ accepts to execute the task from $n_2$. The same is valid for $n_2$. Hence, the rate $a_1$ ($a_2$) at which node $n_1$ ($n_2$) executes tasks from $n_2$ ($n_1$) is: $
a_{1} = \pi_{01} p_1 \lambda_2 \:\:\: a_{2} = \pi_{10} p_2 \lambda_1
$, so that the cooperation ratio of node $i$.: 
$$
R^{in}_1 = a_1 \:\:\:  R^{out}_1 = a_2
$$
$$
R^{in}_2 = a_2 \:\:\:  R^{out}_2 = a_1
$$



\begin{theorem}[Solution of the Minimization Problem $\cP$]
The pair 
$$
p^*_1=1, p^*_2=\Big(\frac{\lambda_2}{\lambda_1}\Big)^2
$$
where $\lambda_1 \ge \lambda_2$, is the solution of $\cP$ under the $MC$ model of Fig.~\ref{fig:MC}.
\end{theorem}
\begin{proof} We must show that this pair: (i) is solution of the MC; (ii) satisfies the constrains of $\cP$; (iii) is Pareto efficient. (i) by plugging these values in (~\ref{equa:sol_pi}), we get $det(A)=1 + \lambda_1 + \lambda_2 + \lambda_1 \lambda_2 + \lambda_2^2 \neq 0$, and hence with this pairs the $MC$ can be solved; in particular, one gets:
$$
\pi_{10}= \frac{\lambda_1}{det(A)}\:\:\: \pi_{01} = \frac{\lambda_2}{det(A)} \:\:\: \pi_{11} = \frac{\lambda_1\lambda_2+\lambda^2_2}{det(A)} 
$$

(ii) from the assumptions, $p^*_2 \le 1$ and hence the pair satisfies (\ref{problem:feasable}) (it is a feasible solution); it also satisfies (\ref{problem:fair}). In fact,
$$
a_1 = \lambda_1 \times \Big (\frac{\lambda_2}{\lambda_1} \Big )^2 \times \frac{\lambda_1}{1 + \lambda_1 + \lambda_2 + \lambda_1 \lambda_2 + \lambda_2^2},
$$
$$
a_2 = \lambda_2 \times \frac{\lambda_2}{1 + \lambda_1 + \lambda_2 + \lambda_1 \lambda_2 + \lambda_2^2}.
$$
Therefore, 
$$
R^{in}_1 = a_1 = a_2 = R^{out}_1 
$$
$$
R^{in}_2 = a_2 = a_1 = R^{out}_2 
$$
Finally, as
$$
b_1(1,p^*_2) = \frac{\lambda_1 \lambda_2 + \lambda^2_2 + \lambda_1 - \frac{\lambda^2_2}{\lambda_1}}{1+\lambda_1+\lambda_2+\lambda_1 \lambda_2 + \lambda^2_2} <  \frac{\lambda_1}{1+\lambda_1} = b^0_1
$$
$$
b_2(1,p^*_2) = \frac{\lambda_1 \lambda_2 + \lambda^2_2}{1+\lambda_1+\lambda_2+\lambda_1 \lambda_2 + \lambda^2_2} <  \frac{\lambda_2}{1+\lambda_2} = b^0_2
$$
it satisfies Equation \ref{problem:convenience} ($b^0_i$ is the value of the Erlang-B formula with just $k=1$ server). (iii) the pair belongs to $\partial P$ and from Property 1 it is a Pareto efficient solution of the relaxed problem of $\cP$ where the above constrains are removed.
\end{proof}

\noindent{\bf Understanding cooperation}
In order to have a better understanding of the nature of the cooperation problem $\cP$, Fig.~\ref{fig:PARETO_09504NEW} shows the $p_1 \times p_2$ domain of the problem with $N=2$, where the $MC$ model with $\lambda_1=0.9,\lambda_2=0.8$ is used to compute $b_1,b_2$.

The solid line at the bottom is the contour line $b_1(p_1,p_2)=b^{0}_1$ and line at the top the one $b_2(p_1,p_2)=b^{0}_2$. The dashed line is the solution set of $\cP$. For a given $p_1$, any $p_2$ above the line at the bottom improves node $n_1$'s blocking probability compared to when $n_1,n_2$ do not cooperate, $b^{0}_1$; however, there is a value of $p_2$ after which $b_2(p_1,p_2)$ becomes higher than $b^{0}_2$, \ie $n_2$ has no benefit to cooperate. The region between the two lines is all the cooperating pairs for which both nodes gains wrt no cooperating for the relaxed problem of $\cP$ when the fair constrain Eq \ref{problem:fair} is removed and the dashed line when this constrain is forced. The two bold lines is the Pareto efficient set for a relaxed problem of $\cP$. Points $P1,P2$ are examples of possible cooperating points for the two nodes when fairness is not considered. $P1$ could roughly represent the cooperation points in a hypothetical agreement between the nodes, when node $n_1$ starts to cooperate first: if node $n_1$ sets $p_1=1$, node $n_2$ can set $p_2$ as indicated by $P1$ as both reduce their $b$. However, the benefit of $n_1$ (percentage reduction of $b_1$) is very low compared to the $n_2$'s one (for example, for $p_2=0.4$ $n_1$ gains just $4\%$, whereas $n_2$ $32\%$). Node $P2$ is a possible cooperation pair when $n_2$ starts to cooperate. The problem of this solution is that the node who started cooperating first is penalized, which may discourage starting to cooperate. Point $P3$ is the solution of $\cP$.

\begin{figure}[b] 
\vspace{-0.6cm}
\centering
  \includegraphics[width=0.79\linewidth]{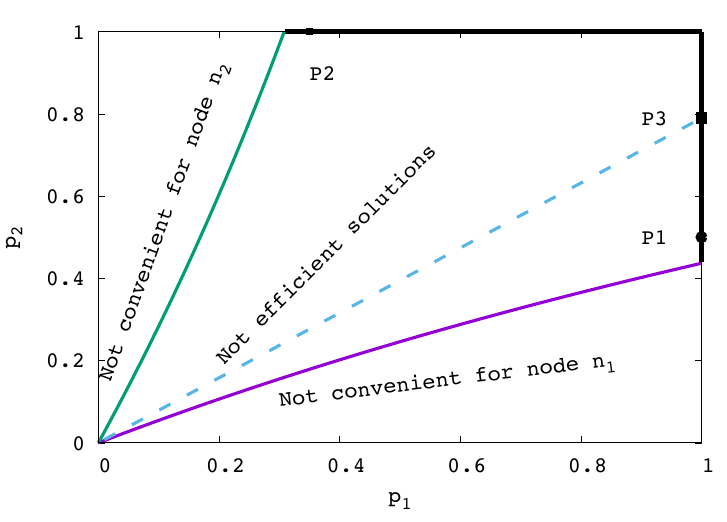}
   \caption{Partition of the cooperation probability domain for $\lambda_1=0.9,\lambda_2=0.8$. }
    \label{fig:PARETO_09504NEW} 
\end{figure}


\subsection{Cooperative Model for N > 2 Nodes}

Cooperation can be extended to $N>2$ nodes as following. When a node is overloaded, the node selects a given node $i$ uniformly at random among the $N-1$, and node $i$ accepts to share its server with probability $p_i$. If the node does not share its server, the original node will send the task to the cloud.
To model this system of cooperating nodes under Poisson arrivals with rates $\lambda_i,i=1,..,N$, we use a $N$-dimensional Markov Chain, $N$-MC. The state of this chain is $s \in S=[0,1]^N$, where $s_i=1$ denotes that the server at node $i$ is overloaded. Let $d(s,s')$ be the Hamming distance between $s$ and $s'$.  For any $s,s' \in S$, the transition rates are defined as following. 

If $d(s,s')=1$ and $s_i=1,s'_i=0$, then:
$$
q_{s,s'}  =1, 
$$
which corresponds to the events of a task leaving node $i$ (dead rates), whereas if $d(s,s')=1$ and $s_i=0,s'_j=1$
$$
q_{s,s'}= \lambda_i + \frac{p_i}{N-1} \sum_{j:s_j=s'_j=1}\lambda_{j}
$$
this rate is the sum of two terms corresponding to two events that make the server of node $i$ to become busy: (\textit{i}) a task arrives from end-users, this occurs with rate $\lambda_i$, or (\textit{ii}) a task arrives to a congested node $j$, the node selects $i$ to probe its availability and $i$ accepts. Finally:
$$
q_{s,s}=- \sum_{s' \ne s} q_{s,s'}.
$$
The above transition rates of this $N$-MC define the infinitesimal generation matrix of the chain of the N-$MC$, $\mathbf{Q}$. The blocking probabilities of the nodes are then computed from the steady state probabilities of $N$-MC as follows: 
$$
b_i= \frac{1}{N-1}\sum_{s:s_i=1,s_j=1} \pi_{s}  + \sum_{s:s_i=1,s_j=0} \pi_s (1-p_j).
$$
A task is in fact blocked if the node receiving the task is overloaded and the selected cooperating node is either overloaded, or it is idle but it does not cooperate. 

The rate at which node $j$ accepts tasks from $i$ is given by:
\begin{equation}
\label{equ:accept}
a_{ij}=\frac{p_j}{N-1} \lambda_i \sum_{s:s_i=1,s_j=0} \pi_{s}.
\end{equation}
In fact, for this event to occur, an arriving task to node $i$ must see node $i$ congested and $i$ idle (the probability of this event is the result of the summation), node $i$ has to select node $j$ (this occurs with probability $\frac{1}{N-1}$) and node $j$ must accept to execute the task ($p_j$). For $N>2$ the rates in Equation \ref{problem:fair} are given by
$$
R^{in}_i = \sum_{j} a_{ji} \:\:\:\: R^{out}_i = \sum_{j} a_{ij}.
$$

\subsection{Numerical solution of $\cP$ for $N>2$ Nodes}
\label{sec:OPT_NUM}
We present a numerical solution of $\cP$ when the blocking probabilities are computed using the $N$-MC model. In the following we assume that $\lambda_i \ge \lambda_{i+1}$. 

First of all, the constrains of (\ref{problem:fair}) are expressed as:
\begin{equation}
\label{equ:fair}
\sum_j a_{ij}-\sum_j a_{ji}=0 \:\:\: \forall i 1,..,N, 
\end{equation}
that can in turn be expressed in matrix form as:
\begin{equation}
\mathbf{F}  \mathbf{p} =0,
\label{equ:fair2}
\end{equation}

where, 
$$
f_{ij} = \lambda_i \sum_{s:s_i=1,s_j=0} \pi_{s}
$$
$$
f_{ii} = -\sum_{j \ne i} f_{ji}= -\sum_{j\ne i} \lambda_j \sum_{s:s_i=0,s_j=1} \pi_{s}. 
$$

Now, following the result of Theorem 1, we conjecture that the solution of the problem is a vector
$$
\mathbf{p}^*=(1,p^*_2,\ldots, p^*_N).
$$

Intuitively, if we take the node $n_1$ point of view, the node sees all the other nodes as a single cooperating entity with normalized load (total load divided number of servers) lower than $\lambda_1$, hence from the solution of $\cP$ for $N=2$, its cooperating probability is to be $p_1=1$.
If this conjecture is true, then $\mathbf{p}^*$ satisfies the set of equations: 
$$
\mathbf{Q} \pi = \mathbf{0}, \mathbf{F} \mathbf{p}^* = \mathbf{0},
$$
where $\mathbf{Q}$ is the N-$MC$ model's generation matrix. And,  to compute such probabilities a fixed point algorithm can be used.


The algorithm first computes the state probabilities for $\mathbf{p}=1$. This set of probabilities exits as in this case the system of nodes is equivalent to an $M/M/N/N$ queue with load $\lambda=\sum_i \lambda_i$. Then, the entries in the $\mathbf{F}$ matrix are computed. From here, a new vector $\mathbf{p}$  of cooperation probabilities is obtained by solving the following system:
$$
\mathbf{F}' \mathbf{p} = [1,\mathbf{0}]^T, 
$$
where $\mathbf{F}'$ is obtained from $\mathbf{F}$ by replacing the first row with $(1,0,\ldots,0)$. In other words, one equation of the system (~\ref{equ:fair2}) is replaced with the equation $p_1=1$. Note that, due to the requirement stated in (\ref{equ:fair}), $\mathbf{F}$ is singular as one row is a linear combination of the all the others.
We exploit the above equation to go back from the constrain (\ref{problem:fair}) to the cooperation vector that would satisfy the above constraints and that $p_1=1$. 
\subsubsection{Numerical evidence of the optimal solution}
Using the above algorithm we find the following numerical evidence of our conjecture on the form of the optimal solution. 

(a) If $\lambda_1$ is not the maximum load, and the remaining $\lambda_i$ are in arbitrary order, the algorithm converged to an unfeasible solution,  \ie $\exists j,  p_j>1$. If these probabilities are normalized to $max \{p_j\}$, the vector $\mathbf{p}^*$ is obtained (after rearranging the index of loads).

(b) If $p_1<1$, the algorithm converged 
a feasible vector $\mathbf{p'}=(p_1 , p^*_2,\ldots,p^*_N)$, which is not optimal. In particular, we have defined the distance from $\mathbf{p}^*$ as: 
$$
dist(p_1) = \sum_i |b_i(\mathbf{p'})-b_i(\mathbf{p}^*)|,
$$
and find that that for $p_1<1$, $dist(p1)>0$. This demonstrates that for any $p_1<1$ a Pareto exits. 







\section{Discussions \& Conclusion}
\label{sec.conc}

This paper has investigated fog-to-fog cooperation within a multi-providers multi-stakeholders scenario. We have proposed a Markov Chain based model and formulate the f2f cooperation as a minimization problem under fairness constraints and assuming uniform cooperative fog node selection. The problem has been solved in a closed form for $N=2$ and numerically for arbitrary $N$ fog nodes. This paper highlights the gain of fog-to-fog cooperation, however this study lacks the following aspects, which we will propose in our future contributions: privacy and security implications of f2f cooperation, the deployment of cooperation in a non homogeneous network, and the optimal fog node cooperating group when $N>2$. We discuss these challenges in the following and we present the future work plans in each of these aspects.

{\it Security \& Privacy}
In this paper, we have assumed a fully collaborative nodes and that all nodes will implement the proposed algorithms. A malicious node, however, can leverage such f2f cooperation scheme to push more tasks and accept fewer or no tasks to execute. This problem can be mitigated in our design as we check the collaboration ratio and the fairness compliance periodically. With respect to privacy concerns, neighboring fog nodes are able to monitor the tasks sent by other providers and their rate which may unveil sensitive business operations/loads. Similar privacy issues have solutions that consists of deploying a proxy that hash all values and make simple comparison of the rates~\cite{alanwar2017proloc}.

{\it Fog Node Coalition}
Our results outline how for $N>2$ providers, the optimal selection of the cooperative fog node may be not uniform and may depend on the loads of fog nodes. Our problem can be generalized to select the optimal list of nodes that can group a coalition set of cooperative nodes. This set can provide the maximum gain among the coalition nodes. The study of this interesting issue left as future work.

\section*{Acknowledgments}
This work was partially supported by US NSF award
\#2148358.

\tiny

\bibliographystyle{IEEEtran}
\bibliography{main}


\end{document}